\documentclass[11pt]{article}
\usepackage[margin=1in]{geometry}
\usepackage{mathtools}
\usepackage{subcaption}
\usepackage{tikz}
\usepackage{amssymb,amsthm}
\usepackage{amsmath}
\usepackage{verbatim}
\usepackage{algorithm}
\usepackage{algorithmic}
\usepackage{hyperref}

   \newtheorem{theorem}{Theorem}
   
    \newtheorem{lemma}[theorem]{Lemma}
   \newtheorem{claim}[theorem]{Claim}
  \DeclareMathOperator{\polylog}{polylog}
  
\title{Maximum Matching on Trees in the Online Preemptive and the Incremental Dynamic Graph Models} \author{Sumedh Tirodkar\\
 \{sumedh.tirodkar@tifr.res.in\}
 \and Sundar Vishwanathan\\
 \{sundar@cse.iitb.ac.in\}}
 \date{}
\begin{document}

\maketitle

\begin{abstract}
We study the Maximum Cardinality Matching (MCM) and the Maximum Weight Matching (MWM) problems, on trees and on some special classes of graphs, in the Online Preemptive and the Incremental Dynamic Graph models. In the {\em Online Preemptive} model, the edges of a graph are revealed one by one and the algorithm is required to always maintain a valid matching. On seeing an edge, the algorithm has to either accept or reject the edge. If accepted, then the adjacent edges are discarded, and all rejections are permanent. In this model, the complexity of the problems is settled for deterministic algorithms~\cite{mcgregor,varadaraja}. Epstein et al.~\cite{epstein} gave a $5.356$-competitive randomized algorithm for MWM, and also proved a lower bound on the competitive ratio of $(1+\ln 2) \approx 1.693$ for MCM. The same lower bound applies for MWM. 

In the {\em Incremental Dynamic Graph} model, at each step an edge is added to the graph, and the algorithm is supposed to quickly update its current matching. Gupta~\cite{gupta} proved that for any $\epsilon\leq 1/2$, there exists an algorithm that maintains a $(1+\epsilon )$-approximate MCM for an incremental bipartite graph in an ``amortized'' update time of $O\left(\frac{\log^2  n}{\epsilon^4}\right)$. No $(2-\epsilon)$-approximation algorithm with a worst case update time of $O(1)$ is known in this model, even for special classes of graphs.

In this paper we show that some of the results can be improved for trees, and for some special classes of graphs. In the online preemptive model, we present a $64/33$-competitive (in expectation) randomized algorithm (which uses only two bits of randomness) for MCM on trees.

Inspired by the above mentioned algorithm for MCM, we present the main result of the paper, a randomized algorithm for MCM with a ``worst case'' update time of $O(1)$, in the incremental dynamic graph model, which is $3/2$-approximate (in expectation) on trees, and $1.8$-approximate (in expectation) on general graphs with maximum degree $3$. Note that this algorithm works only against an oblivious adversary. Hence, we derandomize this algorithm, and give a $(3/2+\epsilon)$-approximate deterministic algorithm for MCM on trees, with an amortized update time of $O(1/\epsilon)$.

We also present a minor result for MWM in the online preemptive model, a $3$-competitive (in expectation) randomized algorithm (that uses only $O(1)$ bits of randomness) on growing trees (where the input revealed upto any stage is always a tree, i.e. a new edge never connects two disconnected trees).
\end{abstract}

\section{Introduction}
The {\em Maximum (Cardinality/Weight) Matching}  problem is one of the most extensively studied problems in Combinatorial Optimization. See Schrijver's book~\cite{schrijver} and references therein for a comprehensive overview of classic work. A {\em matching} $M\subseteq E$ is a set of edges such that at most one edge is incident on any vertex. Traditionally the problem was studied in the offline setting where the entire input is available to the algorithm beforehand. But over the last few decades it has been extensively studied in various other models where the input is revealed in pieces, like the vertex arrival model (adversarial and random), the edge arrival model (adversarial and random), streaming and semi-streaming models, the online preemptive model, etc.~\cite{karp,epstein,epstein2,mcgregor,feigenbaum,pruhs}. In this paper, we study the Maximum Cardinality Matching (MCM) and the Maximum Weight Matching (MWM) problems, on trees and on some special classes of graphs, in the {\em Online Preemptive} model, and in the {\em Incremental Dynamic Graph} model. (Refer Section~\ref{comp} for a comparison between the two models.)

In the online preemptive model, the edges arrive in an online manner, and the algorithm is supposed to accept or reject an edge on arrival. If accepted, the algorithm can reject it later, and all rejections are permanent. The algorithm is supposed to always maintain a valid matching. There is a $5.828$-competitive deterministic algorithm due to McGregor~\cite{mcgregor} for MWM, and a tight lower bound for deterministic algorithms due to Varadaraja~\cite{varadaraja}. Epstein et al.~\cite{epstein} gave a $5.356$-competitive randomized algorithm for MWM, and also proved a $1.693$ lower bound on the competitive ratio achievable by any randomized algorithm for MCM. No better lower bound is known for MWM.


In~\cite{sumedh}, the authors gave the first randomized algorithm with competitive ratio ($28/15$ in expectation) less than $2$ for MCM in the online preemptive model, on growing trees (defined in Section~\ref{prelim}). In Section~\ref{rand_tree}, we extend their algorithm to give a $64/33$-competitive (in expectation) randomized (which uses only two bits of randomness) algorithm  for MCM on trees. Although the algorithm is an extension of the one for growing trees in~\cite{sumedh}, it motivates the algorithm (described in Section~\ref{det_tree}) for MCM in the incremental dynamic graph model.

Note that the adversary presenting the edges in the online preemptive model is oblivious, and does not have access to the random choices made by the algorithm.

In recent years, algorithms for approximate MCM in dynamic graphs have been the focus of many studies due to their wide range of applications. Here~\cite{bernstein,sayan,peng,solomon} is a non-exhaustive list some of the studies. The objective of these dynamic graph algorithms is to efficiently process an online sequence of update operations, such as edge insertions and deletions. It has to quickly maintain an approximate maximum matching despite an adversarial order of edge insertions and deletions.  Dynamic graph problems are usually classified according to the types of updates allowed: incremental models allow only insertions, decremental models allow only deletions, and fully dynamic models allow both. We study MCM in the incremental model. Gupta~\cite{gupta} proved that for any $\epsilon\leq 1/2$, there exists an algorithm that maintains a $(1+\epsilon )$-approximate MCM on bipartite graphs in the incremental model in an ``amortized'' update time of $O\left(\frac{\log^2  n}{\epsilon^4}\right)$. We present a randomized algorithm for MCM in the incremental model with a ``worst case'' update time of $O(1)$, which is $3/2$-approximate (in expectation) on trees, and $1.8$-approximate (in expectation) on general graphs with maximum degree $3$. This algorithm works only against an oblivious adversary. Hence, we derandomize this algorithm, and give a $(3/2+\epsilon)$-approximate deterministic algorithm for MCM on trees, with an amortized update time of $O(1/\epsilon)$. Note that the algorithm of Gupta~\cite{gupta} is based on multiplicative weights update, and it therefore seems unlikely that a better running time analysis for special classes of graphs is possible.

We present a minor result in Section~\ref{algo}, a $3$-competitive (in expectation) randomized  algorithm (which uses only $O(1)$ bits of randomness) for MWM on growing trees in the online preemptive model.  Although, growing trees is a very restricted class of graphs, there are a couple of reasons to study the performance of the algorithm on this class of input. Firstly, almost all lower bounds, including the one due to Varadaraja~\cite{varadaraja} for MWM are on growing trees. Secondly, even for this restricted class, the analysis is involved. We use the primal-dual technique for analyzing the performance of this algorithm, and show that this analysis is indeed tight by giving an example, for which the algorithm achieves the competitive ratio $3$.  We describe the algorithm for general graphs, but are only able to analyze it for growing trees, and new ideas are needed to prove a better bound for general graphs. 
\subsection{Preliminaries}\label{prelim}
We use primal-dual techniques to analyze the performance of all the randomized algorithms described in this paper. Here are the well known Primal and Dual formulations of the matching problem.
 \begin{center}
 \begin{tabular}{c|c}
 Primal LP & Dual LP \\\hline
 $\max \sum_e w_ex_e$ & $\min \sum_v y_v$\\
 $\forall v:\sum_{v\in e}x_e \leq 1$ & $\forall e: y_u + y_v \geq w_e$\\
 $x_e\geq 0$ & $y_v \geq 0$
 \end{tabular}
\end{center}
For MCM, $w_e=1$ for any edge. Any matching $M$ implicitly defines a feasible primal solution. If an edge $e\in M$, then $x_e=1$, otherwise $x_e=0$.

Suppose an algorithm outputs a matching $M$, then let $P$ be the corresponding primal feasible solution. Let $D$ denote some feasible dual solution. The following claim can be easily proved using weak duality.
\begin{claim}\label{ar_lemma}
If $D \leq \alpha\cdot P$, then the algorithm is $\alpha$-competitive.
\end{claim}
If $M$ is any matching, then for an edge $e$, $X(M,e)$ denotes edges in $M$ which share a vertex with the edge $e$. We will say that a vertex(/an edge) is {\em covered } by a matching $M$ if there is an edge in $M$ which is incident on(/adjacent to) the vertex(/edge). We also say that an edge is {\em covered } by a matching $M$ if it belongs to $M$.

In the online preemptive model, growing trees are trees, such that a new edge has exactly one vertex common with already revealed edges.

\subsection{Online Preemptive Model vs. Incremental Dynamic Graph Model}\label{comp}
There are two main differences between these models. Firstly, in the online preemptive model, once an edge is rejected/removed from the matching maintained by the algorithm, it cannot be added into its matching, whereas in the incremental dynamic graph model, rejected/removed edges can be added to the matching later on. Secondly, there is no restriction on how much time an algorithm in the online preemptive model can use to process a revealed edge, whereas in the incremental dynamic graph model, the algorithm is supposed to process the revealed edge fast. The term ``fast'' is used loosely, and is specific to any problem. For example, MCM on general graphs can be found in time $O(m\sqrt{n})$ when the entire input is available~\cite{micali}. But for dynamic graphs, every time an edge is inserted, the algorithm is expected to maintain a matching, approximate if not exact, in time lower than the time required by the optimal offline algorithm for  MCM (say, for instance, in $O(\polylog n)$ amortized time).

\section{MCM in the Online Preemptive Model}\label{rand_tree}
In this section, we present a randomized algorithm (that uses only $2$ bits of randomness) for MCM on trees in the online preemptive model.

The algorithm maintains four matchings $M_1,M_2,M_3,M_4$, and it tries to ensure that a large number of input edges are covered by some or other matchings. (Here, the term ``large number'' is used vaguely. Suppose more than four edges are incident on a vertex, then at most four of them will belong to matchings, one to each.) One of the four matchings is output uniformly at random. A more formal description of the algorithm follows.
   \begin{algorithm}[H]
  \caption{Randomized Algorithm for MCM on Trees}\label{alg_rand1}
   \begin{enumerate}
    \item Pick $l\in_{u.a.r.}\{1,2,3,4\}$.
    \item The algorithm maintains four matchings: $M_1,M_2,M_3,$ and $M_4$.
    \item On arrival of an edge $e$, the processing happens in two phases.
   \begin{enumerate}
	   \item {\bf The Augment phase.} The new edge $e$
		   is added to each $M_i$ in which there are no 
		   edges  adjacent to $e$.
	   \item {\bf The Switching phase.}
		   For $i=2,3,4$, in order, $M_i\gets M_i\setminus X(M_i,e)\cup \{e\}$, provided it decreases the quantity 
 		   $\sum_{j\in[4],i\neq j, |X(M_i\cap M_j,e)| = |X(M_i,e)|}|M_i\cap M_j|$.
   \end{enumerate}
   \item Output  $M_l$.
   \end{enumerate}
 \end{algorithm}
 Although, $l$ is picked randomly at the beginning of the algorithm, this value is not known to the adversary.
 
 Note that in the switching phase, the expected size of the matching stored by the algorithm might decrease. For example, consider two disjoint edges $e_1$ and $e_2$ that have been revealed. Each of them will belong to all four matchings. So the expected size of the matching stored by the algorithm is $2$. Now, if an edge $e$ is revealed between $e_1$ and $e_2$, then $e$ will be added to $M_2$ and $M_3$. The expected size of the matching is now $1.5$. The important thing to notice here is that the decrease is not too much, and we are able to prove that the competitive ratio of the algorithm still remains below $2$.
 
  We begin with the  following observations.
\begin{itemize}
 \item After an edge is revealed, its end points are covered by all four matchings. 
 \item An edge $e$ that does not belong to any matching has four edges incident on its end points such that each of these edges
	 belongs to a distinct matching. This  holds when the edge is revealed, and does not change subsequently.
 \item Every edge is covered by at least three matchings.
\end{itemize}
An edge is called {\em internal} if there are edges incident on both its end points which belong to some matching. An edge is called a {\em leaf edge} either if one of its end point is a leaf or if all the edges incident on one of its end points do not belong to any matching. An edge is called {\em bad} if its end points are covered by only three matchings. 

We begin by proving some properties about the algorithm. The key structural lemma that keeps ``influences'' of bad edges local is given below.  
\begin{lemma}\label{internal}
 At most five consecutive vertices on a path can have bad edges incident on them.
\end{lemma}
According to Lemma~\ref{internal}, there can be at most four consecutive internal bad edges or at most five bad leaf edges incident on five consecutive vertices of a path. Lemma~\ref{internal} is proved in Appendix~\ref{pf_inbad}.

Once all edges have been seen, we distribute the  primal charge among the dual variables, and use the primal-dual framework to infer the competitive ratio. If the end points of every edge are covered with four matchings, then the distribution of dual charge is easy. However we do have bad edges, and would like the edges in matchings to contribute more to the end-points of these edges. Then, the charge on the other end-point would be less and we need to balance this through other edges. Details follow.
\begin{lemma}\label{tree1}
There exists an assignment of the primal charge to the dual variables such that the dual constraint for each edge $e\equiv (u,v)$ is satisfied at least $\frac{33}{64}$ in expectation, i.e. $\mathbb{E}[y_u+y_v]\geq \frac{33}{64}$.
 \end{lemma}
 
\begin{proof}
  Root the tree at an arbitrary vertex. For any edge $e\equiv (u,v)$, let $v$ be the parent vertex, and $u$ be the child vertex. The dual variable assignment is done after the entire input is seen, as follows.\\ 
\textbf{Dual Variable Management:} An edge $e$ will distribute its primal weight between its end-points. The exact values are discussed below. In general, we look to transfer all of the primal charge to the parent vertex. But this does not work and we need a finer strategy. This is detailed below.
\begin{itemize}
 \item If $e$ does not belong to any matching, then it does not contribute to the value of dual variables.
 \item If $e$ belongs to a single matching then, depending on the situation, one of $0$, $\epsilon$, $2\epsilon$, $3\epsilon$, $4\epsilon$, or $5\epsilon$ of its primal charge will be assigned to $u$ and the rest will be assigned to $v$.
 \item If $e$ belongs to two matchings, then at most $6\epsilon$ of its primal charge will be assigned to $u$ as required. The
	 rest is assigned to $v$.
 \item If $e$ belongs to three or four matchings, then its entire primal charge is assigned to $v$.
 \end{itemize}

We will show that $y_u+y_v\geq 2+\epsilon$ for such an edge, when summed over all four matchings. The value of $\epsilon$ is chosen later.  
 
  For the sake of analysis, if there are bad leaf edges incident on both the end points of an internal edge, then we analyze it as a bad internal edge. We need to do this because a bad leaf edge might need to transfer its entire primal charge to the vertex on which there are edges which do not belong to any matching. Note that the end points of the internal edge would still be covered by three matchings, even if we consider that the bad leaf edges do not exist on its end points. The analysis breaks up into eight cases.
  
   \textbf{Case 1.} Suppose $e$ does not belong to any matching. There must be a total of at least $4$ edges incident on $u$ and $v$ besides $e$, each belonging to a distinct matching. Of these $4$, at least a total of $3$, say $e_1$, $e_2$, and $e_3$, must be between some children of $u$ and $v$, to $u$ and $v$ respectively. The edges $e_1$, $e_2$, and $e_3$, each assign a charge of  at least $1-5\epsilon$ to $y_u$ and $y_v$, respectively. Therefore, $y_u+y_v\geq 3-15\epsilon \geq 2+\epsilon$.

   \textbf{Case 2.} Suppose $e$ is a bad leaf edge that belongs to a single matching, and internal edges are incident on $v$. This implies that there is an edge $e_1$ from a child vertex of $v$ to $v$, which belongs to single matching, and another edge $e_2$, also belonging to single matching from $v$ to its parent vertex. The edge $e$ assigns a charge of $1$ to $y_v$. If $e_1$ assigns a charge of $1$ or $1-\epsilon$ or $1-2\epsilon$ or $1-3\epsilon$  or $1-4\epsilon$ to $y_v$, then $e_2$ assigns $\epsilon$ or $2\epsilon$ or $3\epsilon$ or $4\epsilon$ or $5\epsilon$ respectively to $y_v$.  In either case, $y_u+y_v=2+\epsilon$. The key fact is  that $e_1$ could not have assigned $5\epsilon$ to its child vertex. Since, then, by Lemma~\ref{internal}, $e$ cannot be a bad edge.

   \textbf{Case 3.} Suppose $e$ is a bad leaf edge that belongs to a single matching, and internal edges are incident on $u$. This implies that there are two edges $e_1$ and $e_2$ from children of $u$ to $u$, each belonging to a single distinct matching. The edge $e$ assigns a charge of $1$ to $y_v$. Both $e_1$ and $e_2$ assign a charge of at least $1-4\epsilon$ to $y_u$. In either case, $y_u+y_v\geq 3-8\epsilon \geq 2+\epsilon$. The key fact is that neither $e_1$ nor $e_2$ could have assigned more than $4\epsilon$ to their corresponding child vertices. Since, then, by Lemma~\ref{internal}, $e$ cannot be a bad edge.

   \textbf{Case 4.} Suppose $e$ is an internal bad edge. This implies (by Lemma~\ref{internal}) that there is an edge  $e_1$ from a child vertex of $u$ to $u$, which belongs to a single  matching. Also, there is an edge $e_2$, from $v$ to its parent vertex (or from a child vertex $v$ to $v$), which also belongs to a single matching. The edge $e$ assigns its remaining charge ($1$ or $1-\epsilon$ or $1-2\epsilon$ or $1-3\epsilon$  or $1-4\epsilon$) to $y_v$. If $e_1$ assigns a charge of $1$ or $1-\epsilon$ or $1-2\epsilon$ or $1-3\epsilon$  or $1-4\epsilon$ to $y_u$, then $e_2$ assigns $\epsilon$ or $2\epsilon$ or $3\epsilon$ or $4\epsilon$ or $5\epsilon$ respectively to $y_v$.  In either case, $y_u+y_v=2+\epsilon$. The key fact is  that $e_1$ could not have assigned $5\epsilon$ to its child vertex. Since, then, by Lemma~\ref{internal}, $e$ cannot be a bad edge.
  
  \textbf{Case 5.} Suppose $e$ is not a bad edge, and it belongs to a single matching. Then either there are at least two edges $e_1$ and $e_2$ from child vertices of $u$ or $v$ to $u$ or $v$ respectively, or $e_1$ on $u$ and $e_2$ on $v$, each belonging to a single matching, or one edge $e_3$ from a child vertex of $u$ or $v$ to $u$ or $v$, respectively, which belongs to two matchings, or one edge $e_4$ from a child vertex of $u$ or $v$ to $u$ or $v$, respectively, which belongs to single matching, and one edge $e_5$  from $v$ to its parent vertex which belongs to two matchings. In either case, $y_u+y_v\geq 3-10\epsilon\geq 2+\epsilon$.
    
  \textbf{Case 6.} Suppose $e$ is a bad edge that belongs to two matchings, and internal edge is incident on $u$ or $v$.
	  This implies that there is an edge $e_1$, from a child vertex of $u$ to $u$ or from $v$ to its parent vertex which belongs to a single matching. The edge $e$ assigns a charge of $2$ to $y_v$, and the edge $e_1$  assigns a charge of $\epsilon$ to $y_u$ or $y_v$ respectively. Thus, $y_u+y_v=2+\epsilon$. 
  
  \textbf{Case 7.} Suppose $e$ is not a bad edge and it belongs to two matchings. This means that  either there is an edge $e_1$ from a child vertex of $u$ to $u$, which belongs to at least one matching,  or there is an edge from child vertex of $v$ to $v$ that belongs to at least one matching, or there is an edge from $v$ to its parent vertex which belongs to two matchings. The edge $e$ assigns a charge of  $2$ among $y_u$ and $y_v$. The neighboring edges assign a charge of $\epsilon$ to $y_u$ or $y_v$ (depending on which vertex it is incident to), yielding  $y_u+y_v\geq 2+\epsilon$.
	  
  \textbf{Case 8.} Suppose, $e$ belongs to $3$ or $4$ matchings, then trivially $y_u+y_v\geq 2+\epsilon$.
  
  From the above cases, $y_v+y_v\geq 3-15\epsilon$ and $y_u+y_v\geq 2+\epsilon$. The best value for the competitive ratio is obtained when $\epsilon=\frac{1}{16}$, yielding $\mathbb{E}[y_u+y_v]\geq \frac{33}{64}$.
\end{proof}
Lemma~\ref{tree1} immediately implies Theorem~\ref{ub_MCM} using Claim~\ref{ar_lemma}.
\begin{theorem}\label{ub_MCM}
 Algorithm~\ref{alg_rand1} is a $\frac{64}{33}$-competitive randomized algorithm for finding MCM on trees.
\end{theorem}

\section{MCM in the Incremental Dynamic Graph Model}\label{det_tree}
In this section, we present our main result, a randomized algorithm (that uses only $O(1)$ bits of randomness) for MCM in the incremental dynamic graph model, which is $3/2$-approximate (in expectation) on trees, and is $1.8$-approximate (in expectation) on general graphs with maximum degree $3$, with $O(1)$ worst case update time per edge. It is inspired by the randomized algorithm for MCM on trees described in Section~\ref{rand_tree}. In the online preemptive model, we cannot add edges in the matching which were discarded earlier, which results in the existence of bad edges. But in the incremental dynamic graph model, there is no such restriction. For some $i\in [3]$, let $e\equiv (u,v) \in M_i$ be switched out by some edge $e'\equiv (u,u')$, i.e. $M_i\gets  M_i\setminus \{e\}\cup\{e'\}$. If there is an edge $e''\equiv (v,v')\in M_j$ for $i\neq j$, then we can add $e''$ to $M_i$ if possible. Using this simple trick, we get a better approximation ratio in this model, and also, the analysis  becomes significantly simpler. Details follow.

   \begin{algorithm}[H]
   \caption{Randomized Algorithm for MCM}\label{alg_rand2}
   \begin{enumerate}
    \item Pick $l \in_{u.a.r.} \{1,2,3\}$.
    \item The algorithm maintains three matchings: $M_1,M_2,$ and $M_3$.
    \item When an edge $e$ is inserted, the processing happens in two phases.
   \begin{enumerate}
	   \item {\bf The Augment phase.} The new edge $e$
		   is added to each $M_i$ in which there are no 
		   edges  adjacent to $e$.
	   \item {\bf The Switching phase.}
		   For $i=2,3$, in order, $M_i\gets M_i\setminus X(M_i,e)\cup \{e\}$, provided it decreases the quantity 
 		   $\sum_{j\in[3],i\neq j, |X(M_i\cap M_j,e)| = |X(M_i,e)|}|M_i\cap M_j|$.\\
		   For every edge $e'$ discarded from $M_i$, add edges on the other end point of $e'$ in $M_j$ ($\forall j\neq i$) to $M_i$ if possible.
   \end{enumerate}
   \item Output the matching $M_l$ on query.
   \end{enumerate}
 \end{algorithm}
 
 Note that the end points of every edge will be covered by all three matchings, and hence all three matchings are maximal. 
 
We again use the primal-dual technique to analyze the performance of this algorithm on trees. 
\begin{lemma}\label{tree2}
There exists an assignment of the primal charge amongst the dual variables such that the dual constraint for each edge $e\equiv (u,v)$ is satisfied at least $\frac{2}{3}$rd in expectation.
 \end{lemma}
\begin{proof}
Root the tree at an arbitrary vertex. For any edge $e\equiv (u,v)$, let $v$ be the parent vertex, and $u$ be the child vertex. The dual variable assignment is done at the end of input/on query, as follows.
\begin{itemize}
 \item If $e$ does not belong to any matching, then it does not contribute to the value of dual variables.
 \item If $e$ belongs to a single matching, then its entire primal charge is assigned to $v$ as $y_v=1$.
 \item If $e$ belongs to two matchings, then its entire primal charge is assigned equally amongst $u$ and $v$, as $y_u=1$ and $y_v=1$.
 \item If $e$ belongs to three matchings, then its entire primal charge is assigned to $v$ as $y_v=3$.
 \end{itemize}
 The analysis breaks up into three cases.\\  
   \textbf{Case 1.} Suppose $e$ does not belong to any matching. There must be a total of at least $2$ edges incident amongst $u$ and $v$ besides $e$, each belonging to a distinct matchings, from their respective children. Therefore, $y_u+y_v\geq 2$.\\
     \textbf{Case 2.} Suppose $e$ belongs to a single matching. Then either there is an edge $e'$ incident on $u$ or $v$ which belongs to a single matching, from their respective children, or there is an edge $e''$ incident on $u$ or $v$ which belongs to two matchings. In either case, $y_u+y_v\geq 2$.\\
   \textbf{Case 3.} Suppose $e$ belongs to two or three matchings, then $y_u+y_v\geq 2$ trivially.
\end{proof}
Lemma~\ref{tree2} immediately implies Theorem~\ref{ub_MCM2} using Claim~\ref{ar_lemma}.
\begin{theorem}\label{ub_MCM2}
 Algorithm~\ref{alg_rand2} is a $\frac{3}{2}$-approximate (in expectation) randomized algorithm for MCM on trees, with a worst case update time of $O(1)$.
\end{theorem}
We also analyze Algorithm~\ref{alg_rand2} for general graphs with maximum degree $3$, and prove the following Theorem.
\begin{theorem}\label{thm_gen3}
 Algorithm~\ref{alg_rand2} is a $1.8$-approximate (in expectation) randomized algorithm for MCM on general graphs with maximum degree $3$, with a worst case update time of $O(1)$.
\end{theorem}
We use the following Lemma to prove Theorem~\ref{thm_gen3}.

\begin{lemma}\label{gen3}
There exists an assignment of the primal charge amongst the dual variables such that the dual constraint for each edge $e\equiv (u,v)$ is satisfied at least $\frac{5}{9}$th in expectation.
 \end{lemma}
\begin{proof}
The dual variable assignment is done at the end of input/or query, as follows.
\begin{itemize}
 \item If $e$ does not belong to any matching, then it does not contribute to the value of dual variables.
 \item If $e$ belongs to a single matching, then there are two sub cases.
 \begin{enumerate}
  \item W.l.o.g., if $u$ is covered by a single matching, then  primal charge $x_e=1$ is divided as $y_u=1/2+\epsilon$ and $y_v=y_v+1/2-\epsilon$.
  \item If both $u$ and $v$ are covered by at least two matchings, then  primal charge $x_e=1$ is divided as $y_u=y_u+1/2$ and $y_v=y_v+1/2$.
 \end{enumerate}
 \item If $e$ belongs to two or three matchings, then its entire primal charge is divided equally amongst $u$ and $v$.
 \end{itemize}
 The analysis breaks up into three cases.\\
   \textbf{Case 1.} Suppose $e$ does not belong to any matching. Then $u$ and $v$ must be covered by a total of at least $3$ matchings (counting multiplicities). W.l.o.g., if $u$ is covered by a single matching, then $v$ has to be covered by at least two matchings. Hence, $y_u=1/2+\epsilon$, and $y_v\geq 1$. Else, both $u$ and $v$ are covered by at least two matchings, then $y_u\geq 1$ and $y_v\geq 1$. Therefore, $y_u+y_v\geq 3/2+\epsilon$.\\
     \textbf{Case 2.} Suppose $e$ belongs to a single matching. Then, $y_u+y_v\geq 1+1/2-\epsilon +1/2-\epsilon =2-2\epsilon \geq 3/2+\epsilon$.\\
   \textbf{Case 3.} Suppose $e$ belongs to two or three matchings, then $y_u+y_v\geq 3/2+\epsilon$ trivially.

   The proof of Lemma is complete with $\epsilon=1/6$.   
\end{proof}
Lemma~\ref{gen3} immediately implies Theorem~\ref{thm_gen3} using Claim~\ref{ar_lemma}.

\subsection{A Deterministic Algorithm}
Note that Algorithm~\ref{alg_rand2} only works against an oblivious adversary. In this section, we derandomize Algorithm~\ref{alg_rand2} to give a $(3/2+\epsilon)$-approximation deterministic algorithm, for MCM on trees, with an amortized update time of $O(1/\epsilon)$, for any $\epsilon \leq 1/2$.

   \begin{algorithm}[H]
   \caption{Deterministic Algorithm for MCM}\label{alg_det1}
   \begin{enumerate}
    \item Let $\epsilon \in (0,1/2]$ be some input parameter, and $c=1$.
    \item The algorithm maintains four matchings: $M_1,M_2,M_3$ and a support matching $M_4$. On query, output matching $M_c$.
    \item When an edge $e$ is inserted, the processing happens in four phases.
   \begin{enumerate}
	   \item {\bf The Augment phase.} The new edge $e$
		   is added to each $M_i$ in which there are no 
		   edges  adjacent to $e$.
	   \item {\bf The Switching phase.}
		   For $i=2,3$, in order, $M_i\gets M_i\setminus X(M_i,e)\cup \{e\}$, provided $|X(M_i,e)|=1$ and it decreases the quantity 
 		   $\sum_{j\in[3],i\neq j}|M_i\cap M_j|$.\\
		   For the edge $e'$ that is discarded from $M_i$, add edges on the other end point of $e'$ in $M_j$ ($\forall j\in [4],j\neq i$) to $M_i$ if possible.
	    \item {\bf The Support phase.} If the edge $e$ was not added to any $M_i$, $\forall i\in[3]$, in the Augment or Switching phase, then add it to the support matching $M_4$ if there are no edges adjacent to it in $M_4$.
	    \item {\bf The ChangeCurr phase.} If $|M_c|< \left(|M_i|+|M_j|\right)/(2(1+\epsilon))$ such that $i,j,c\in[3]$, and are all distinct, then set $c = k$ if $M_k$ is the matching of maximum size among $M_1,M_2,M_3$.
   \end{enumerate}
   \end{enumerate}
 \end{algorithm}
Note that there are two more phases in this algorithm than Algorithm~\ref{alg_rand2}, and there is also a minor modification in the description of switching phase. These changes are done to ensure that the size of any matching maintained by the algorithm never decreases(, which helps with the analysis as pointed out later). An edge can be added to $M_i$ in the Switching phase only if it has one conflicting edge in $M_i$. But this can result in $M_2$ and $M_3$ not being maximal matchings(, which is again required in the analysis as pointed out later). The only way this can happen is if some edge $e$ is not added to any matching in the Augment phase, and later on, after the Switching phase, its end points are not covered by $M_2$ and $M_3$. We add such an edge to the support matching $M_4$, and this edge is later added to $M_2$ and $M_3$ in the Switching phase, thereby ensuring their maximality. With these modifications, the approximation ratio claimed in Theorem~\ref{ub_MCM2} still holds on average size of matchings $M_1,M_2,M_3$ stored by this algorithm.

We prove the following theorem for Algorithm~\ref{alg_det1}.
 \begin{theorem}\label{ub_MCM3}
 Algorithm~\ref{alg_det1} is $\left(\frac{3}{2}+\epsilon\right)$-approximate for MCM on trees, with an amortized update time of $O(1/\epsilon)$.
\end{theorem}
\begin{proof}
 We first prove the approximation ratio, and then argue about the update time per edge.
 
 Step (3c) in Algorithm~\ref{alg_det1} ensures that at each stage $|M_c|\geq \left(|M_i|+|M_j|\right)/(2(1+\epsilon))$, such that $i,j,c \in [3]$, and are all distinct, and $M_c$ is the current matching which will be output by the algorithm on query. Let $M$ be the optimum matching at any stage. Theorem~\ref{ub_MCM2} implies that
 \begin{align*}
 \frac{|M_c|+|M_i|+|M_j|}{3} &\geq \frac{2}{3} |M|\\
 \Longrightarrow |M_c|+2(1+\epsilon)|M_c| &\geq 2|M|\\
 \Longrightarrow\left(\frac{3}{2}+\epsilon\right)|M_c| &\geq |M|.
 \end{align*}
 Note that the approximation ratio trivially holds after the first edge is inserted as we set $c=1$(, and will hold even if we set $c=2$ or $3$, because first edge is added to all three matchings $M_1,M_2,M_3$).
 
 In the augment or the switching phase, $O(1)$ time is spent per edge. Let $M'', M_c'',M_i'',M_j''$ represent the matchings $M,M_c,M_i,M_j$ immediately after $c$ was updated, and  let $M', M_c',M_i',M_j'$ represent the respective matchings immediately after the previous time $c$ was updated. In the ChangeCurr phase, at most $2|M''|$ time is potentially spent (because size of any matching stored by the algorithm is at most $|M''|$), while changing the current matching output by the algorithm. But we show that this happens very rarely. If $|M''|\geq 2|M'|$, then at least $|M''|/2$ edges have been inserted between two recent updates of $c$. This implies an amortized update time of $O(1)$ per edge.
 
 Now suppose $|M''| <2|M'|$. Immediately after the previous update of $c$, $|M_c'|\geq (|M_i'|+|M_j'|)/2$ (because $M_c$ is the maximum size matching among $M_1,M_2,$ and $M_3$). Just before $c$ is updated in the ChangeCurr phase, $|M_c''| < (|M_i''|+|M_j''|)/(2(1+\epsilon))$. So, the change in the value of $|M_c|$ is at most
 \[
  \frac{|M_i''|+|M_j''|}{2(1+\epsilon)} - \frac{|M_i'|+|M_j'|}{2}.
 \]
But this value is at least zero, as the size of any matching can never decrease by the description of the algorithm. Hence,
\begin{align*}
 \frac{|M_i''|+|M_j''|}{2(1+\epsilon)} - \frac{|M_i'|+|M_j'|}{2} &\geq 0 \\
 \Longrightarrow |M_i''|+|M_j''| - (1+\epsilon)(|M_i'|+|M_j'|) &\geq 0 \\
 \Longrightarrow (|M_i''|-|M_i'|)+(|M_j''|-|M_j'|) &\geq \epsilon(|M_i'|+|M_j'|) \\
 \Longrightarrow (|M_i''|-|M_i'|)+(|M_j''|-|M_j'|) &\geq \epsilon|M'| &\dots M_i',M_j'\text{ are maximal} .
\end{align*}
Thus, along with the fact that $|M''| <2|M'|$, $\Omega(\epsilon|M''|)$ edges have been inserted between two recent updates of value of $c$. This implies an amortized update time of $O(1/\epsilon)$ per edge, and finishes the proof.

\end{proof}

\section{MWM in the Online Preemptive Model}\label{algo}
In this section, we present a randomized algorithm (that uses only $O(1)$ bits of randomness) for MWM in the online preemptive model, and analyze its performance for growing trees. The algorithm is motivated by the deterministic algorithm for MWM due to McGregor~\cite{mcgregor}. McGregor's algorithm is easy to describe -- if the weight of the new edge is more than $(1+\gamma)$ times the weight of the conflicting edges in the current matching, then evict them and add the new edge. The algorithm is $(1+\gamma)(2+1/\gamma)$-competitive, and attains the best competitive ratio of $3+2\sqrt{2} \approx 5.828$ for $\gamma=\frac{1}{\sqrt{2}}$. It achieves this competitive ratio for the following example. Start by presenting an edge of weight $x_0=1$ to the algorithm. This edge will be added to the matching. Assume inductively that after iteration $i$, the algorithm's matching has only the edge of weight $x_i$. In iteration $i+1$, present an edge of weight $y_{i+1}=(1+\gamma)x_i$ on one end point of $x_i$ (we slightly abuse the notation here, and say that $x_i$ is also the name of the edge of weight $x_i$). This edge will not be accepted in the algorithm's matching. Give an edge of weight $x_{i+1}=(1+\gamma)x_i + \epsilon$ on the other end point of $x_i$. This edge will be accepted in the algorithm's matching, and $x_i$ will be evicted. This process terminates for some large $n$, letting $x_{n+1}=(1+\gamma)x_n$. The edge of weight $x_{n+1}$ will not be accepted in the algorithm's matching. The algorithm will hold only the edge of weight $x_n$, whereas the optimum matching would include edges of weight $y_1,\dots,y_{n+1},x_{n+1}$. It can be easily inferred that this gives the required lower bound on the competitive ratio.
  
Notice that the edges presented in the example crucially depended on $\gamma$. To beat this, we maintain two matching, with $\gamma$ values $\gamma_1$ and $\gamma_2$ respectively, and choose one at random. We describe the algorithm next.
  
 \begin{algorithm}[H]
  \caption{Randomized Algorithm for MWM}\label{algo_mwm}
  \label{MWM_rand_algo}
  \begin{enumerate}
   \item Maintain two matchings $M_1$ and $M_2$. Let $j=1$ with probability $p$, and $j=2$ otherwise.
   \item On receipt of an edge $e$:\\
	  For $i=1,2$, if $w(e)>(1+\gamma_i)w(X(M_i,e))$, then $M_i=M_i\setminus X(M_i,e) \cup \{e\}$.
   \item Output $M_j$.
  \end{enumerate}

 \end{algorithm}
 
 Note that we cannot just output the best of two matchings because that could violate the constraints of the online preemptive model.
\subsection{Analysis}
We use the primal-dual technique to analyze the performance of this algorithm. The primal-dual technique used to analyze McGregor's deterministic algorithm for MWM described in~\cite{sumedh} is fairly straightforward. However the management becomes complicated with the introduction of randomness, and we are only able to analyze the algorithm in a very restricted class of graphs, which are growing trees. 

\begin{theorem}\label{MWM_rand}
The expected competitive ratio of Algorithm~\ref{algo_mwm} on growing trees is
\[
\max\left\{\frac{1+\gamma_1}{p},\frac{1+\gamma_2}{1-p},\frac{(1+\gamma_1)(1+\gamma_2)(1+2\gamma_1)}{p\cdot \gamma_1 + (1-p)\gamma_2 + \gamma_1\gamma_2}\right\},
\]
where $p$ is the probability to output $M_1$.
\end{theorem}

We maintain both primal and dual variables along with the run of the algorithm. Consider a round in which an edge $e\equiv (u,v)$ is revealed, where $v$ is the new vertex. Before $e$ is revealed, let $e_1$ and $e_2$ be the edges incident on $u$ which belong to $M_1$ and $M_2$ respectively. If such an $e_i$ does not exist, then we may assume $w(e_i)=0$. The primal and dual variables are updated as follows.
\begin{itemize}
\item $e$ is rejected by both matchings, we set the primal variable $x_e=0$, and the dual variable $y_v=0$.
 \item $e$ is added to $M_1$ only, then we set the primal variable $x_e=p$, and the dual variable $y_u=\max(y_u,\min((1+\gamma_1) w(e), (1+\gamma_2) w(e_2)))$, and $y_v=0$;. 
 \item $e$ is added to $M_2$ only, then we set the primal variable $x_e=1-p$, and the dual variable $y_u=\max(y_u,\min((1+\gamma_1)w(e_1),(1+\gamma_2)w(e)))$, and $y_v=0$. 
 \item $e$ is added to both the matchings, then we set the primal variable $x_e=1$, and the dual variables $y_u=\max(y_u,(1+\gamma_1)w(e))$ and $y_v=(1+\gamma_1)w(e)$. 
 \item When an edge $e'$ is evicted from $M_1$ (or $M_2$), we decrease its primal variable $x_{e'}$ by $p$ (or $(1-p)$ respectively), and the corresponding dual variables are unchanged.
\end{itemize}
We begin with three simple observations.
\begin{enumerate}
 \item The cost of the primal solution is equal to the expected weight of the matching maintained by the algorithm.
 \item  The dual variables never decrease. Hence, if a dual constraint is feasible once, it remains so.
 \item $y_u \geq \min((1+\gamma_1)w(e_1),(1+\gamma_2)w(e_2))$.
\end{enumerate}
The idea behind the analysis is to prove a bound on the ratio of the dual cost and the primal cost while maintaining dual feasibility. By Observation $2$, to ensure dual feasibility, it is sufficient to ensure feasibility of the dual constraint of the new edge. If the new edge $e$ is not accepted in any $M_i$, then $w(e)\leq \min((1+\gamma_1)w(e_1),(1+\gamma_2)w(e_2))$. Hence, the dual constraint is satisfied by Observation $3$. Else, it can be seen that  the dual constraint is satisfied by the updates performed on the dual variables.

The following lemma implies Theorem~\ref{MWM_rand} using Claim~\ref{ar_lemma}.
\begin{lemma}\label{lemma_MWM}
 $\frac{\Delta \text{Dual}}{\Delta\text{Primal}} \leq \max\left\{\frac{1+\gamma_1}{p},\frac{1+\gamma_2}{1-p},\frac{(1+\gamma_1)(1+\gamma_2)(1+2\gamma_1)}{p\cdot \gamma_1 + (1-p)\gamma_2 + \gamma_1\gamma_2}\right\}$\\ after every round.
\end{lemma}
We will use the following simple technical lemma to prove Lemma~\ref{lemma_MWM}.
\begin{lemma}\label{key_lemma}
 $\frac{ax+b}{cx+d}$increases with $x$ iff $ad-bc\geq 0$.
\end{lemma}
\begin{proof}[Proof of Lemma~\ref{lemma_MWM}]
There are four cases to be considered.
\begin{enumerate}
 \item If edge $e$ is accepted in $M_1$, but not in $M_2$. Then $(1+\gamma_1)w(e_1)< w(e) \leq (1+\gamma_2)w(e_2)$. By Observation $3$, before $e$ was revealed, $y_u\geq(1+\gamma_1)w(e_1)$. After $e$ is accepted in $M_1$, $\Delta\text{Primal} =p(w(e)-w(e_1))$, and $\Delta\text{Dual} \leq (1+\gamma_1)(w(e)-w(e_1))$. Hence,
 \[
  \frac{\Delta\text{Dual}}{\Delta\text{Primal} }\leq \frac{(1+\gamma_1)}{p}.
 \]
\item If edge $e$ is accepted in $M_2$, but not in $M_1$. Then $(1+\gamma_2)w(e_2)< w(e) \leq (1+\gamma_1)w(e_1)$. By Observation $3$, before $e$ was revealed, $y_u\geq(1+\gamma_2)w(e_2)$. After $e$ is accepted in $M_2$, $\Delta\text{Primal} =(1-p)(w(e)-w(e_2))$, and $\Delta\text{Dual} \leq (1+\gamma_2)(w(e)-w(e_2))$. Hence,
 \[
  \frac{\Delta\text{Dual}}{\Delta\text{Primal} }\leq \frac{(1+\gamma_2)}{1-p}.
 \]

 \item If edge $e$ is accepted in both the matchings, and $(1+\gamma_1)w(e_1) \leq (1+\gamma_2)w(e_2)$ $ < w(e)$. By Observation $3$, before $e$ was revealed, $y_u\geq(1+\gamma_1)w(e_1)$. After $e$ is accepted in both the matchings, $\Delta\text{Dual} \leq (1+\gamma_1)(2w(e)-w(e_1))$. The change in primal cost is
 \begin{align*}
 \Delta\text{Primal} &\geq w(e)-p\cdot w(e_1) - (1-p)\cdot w(e_2) \\
  &\geq w(e) - p\cdot w(e_1) - (1-p)\cdot \frac{w(e)}{1+\gamma_2}\\
  &=\frac{p+\gamma_2}{1+\gamma_2}w(e) -p\cdot w(e_1).\\
  \frac{\Delta\text{Dual}}{\Delta\text{Primal} }&\leq (1+\gamma_1)\frac{2w(e)-w(e_1)}{\frac{p+\gamma_2}{1+\gamma_2}w(e) -p\cdot w(e_1)}.  
 \end{align*}
By Lemma~\ref{key_lemma}, this value increases, for a fixed $w(e)$, with $w(e_1)$ if $\gamma_2\leq \frac{p}{1-2p}$, and its worst case value is achieved when $(1+\gamma_1)w(e_1)=w(e)$. Thus,
\begin{align*}
 \frac{\Delta\text{Dual}}{\Delta\text{Primal} }&\leq (1+\gamma_1)\frac{2(1+\gamma_1)(1+\gamma_2)-(1+\gamma_2)}{(p+\gamma_2)(1+\gamma_1) -p(1+\gamma_2)}\\
 &= (1+\gamma_1)(1+\gamma_2)\frac{1+2\gamma_1}{p\cdot \gamma_1 + (1-p)\gamma_2 + \gamma_1\gamma_2}.
\end{align*}

\item If $e$ is accepted in both the matchings, and $(1+\gamma_2)w(e_2) \leq (1+\gamma_1)w(e_1) < w(e)$. By Observation $3$, before $e$ was revealed, $y_u\geq(1+\gamma_2)w(e_2)$. The following bound can be proved similarly.
\begin{align*}
 \frac{\Delta\text{Dual}}{\Delta\text{Primal} } &\leq (1+\gamma_1)(1+\gamma_2)\frac{1+2\gamma_1}{p\cdot \gamma_1 + (1-p)\gamma_2 + \gamma_1\gamma_2}.
\end{align*}
\end{enumerate}
\end{proof}
The following theorem is an immediate consequence of Theorem~\ref{MWM_rand}.
\begin{theorem}\label{ub_MWM}
Algorithm~\ref{algo_mwm} is a $3$-competitive (in expectation) randomized algorithm for MWM on growing trees, when $p=1/3$, $\gamma_1=0$, and $\gamma_2=1$; and the analysis is tight.
\end{theorem}
The input for which Algorithm~\ref{algo_mwm} is $3$-competitive is as follows. Start by presenting an edge of weight $x_0=1$. It will be added to both $M_1$ and $M_2$. Assume inductively, that currently both matching only contain an edge of weight $x_i$. Present an edge of weight $y_{i+1}=x_i$ on one end point of $x_i$. This edge will not be accepted in either of the matchings. Present an edge of weight $x_{i+1}=2\cdot x_i + \epsilon$ on the other end point of $x_i$. This edge will be accepted in  both the  matchings, and $x_i$ will be evicted. For a sufficiently large value $n$, let $x_{n+1}=x_n$. So edge of weight $x_{n+1}$ will not be accepted in either of the matchings. Both the matchings will hold only the edge of weight $x_n$, whereas the optimum matching would include edges of weight $y_1,\dots,y_{n+1},x_{n+1}$. The weight of the matching stored by the algorithm is $2^n$, whereas the weight of the optimum matching is $\approx 3\cdot 2^n$ (we have ignored the $\epsilon$ terms here). This gives the competitive ratio $3$.

\textbf{Note.} In the analysis of Algorithm~\ref{algo_mwm} for growing trees, we crucially use the  following fact in the dual variable assignment. If an edge $e\notin M_i$ for some $i$, then a new edge incident on its leaf vertex will definitely be added to $M_i$, and it suffices to assign a zero charge to the corresponding dual variable. This is not necessarily true for more general classes of graphs, and new ideas are needed to analyze the performance for those classes.

\section*{Acknowledgements}
The first author would like to thank Ashish Chiplunkar for helpful suggestions to improve the competitive ratio of Algorithm~\ref{algo_mwm}, and also to improve the presentation of Section~\ref{algo}.

\bibliographystyle{plainurl}
\bibliography{paper}
\begin{appendix}

 \section{Proof of Lemma~\ref{internal}}\label{pf_inbad}
 We crucially use the following lemma to prove Lemma~\ref{internal}.
 \begin{lemma}\label{M4bad}
  (a) If an edge $e$ belongs only to $M_4$ at the end of input, then bad edges cannot be incident on both its end points.\\
  (b) Also, if an edge $e$ was added to $M_4$ only in the switching phase, then $e$ cannot be a bad edge.
 \end{lemma}
 \begin{proof}
  There are two cases to consider.
  \begin{enumerate}
   \item Suppose $e$ was added to $M_4$ only when it was revealed. Then on one of its end point either there should be two edges incident (other than $e$), such that each of them belongs to a single matching, or there should be one edge which belongs to two matchings. In either case, the edges incident on that end point of $e$ should have neighboring edges which belong to some matching (by description of algorithm). And hence, these edges cannot be bad.
   \item Suppose $e$ was added to $M_4$ as well as some other matching when it was revealed. If $e$ belonged to three matchings when it was revealed, then its neighboring edge will have its end points covered by at least four matching edges, and this number can never go below four. If $e$ belonged to two matchings when it was revealed, then it\\
   -- (a) either has one neighboring edge which belongs to two matchings,\\
   -- (b) or one neighboring edge on each of its end points, each belonging to distinct matching,\\
   -- (c) or two neighboring edges on one of its end points, such that both of them belong to distinct matchings.\\
   In Case (a), this neighboring edge should have a neighboring edge on its other end point which belongs to some matching, and hence it cannot be a bad edge. In Case (b), each of these edge should have at least two neighboring edges of their own on their respective other end point, which belong to certain matching. Hence, both these edges cannot be bad. In Case (c), both these edges should have neighboring edges of their own on their respective other end point, which belong to certain matching. Hence, both these edges cannot be bad.
  \end{enumerate}
For the second part of lemma, if edge $e$ added to $M_4$ in the switching phase, then it means that $e$ will have three neighboring edges $e_1$,$e_2$, and $e_3$, belonging to $M_1$, $M_2$, and $M_3$, respectively. This is because $e$ will be added to $M_4$ in the switching phase only if it is not added to $M_2$ or $M_3$ in the switching phase, which means there are edges which belong only to $M_2$ and $M_3$ respectively.
 \end{proof}
\begin{proof}[Proof of Lemma~\ref{internal}]
There are two cases to consider.
\begin{enumerate}
 \item Suppose if there is a bad leaf edge $e$ which belongs to $M_4$. If $e$ is added to $M_4$ in the switching phase, then $e$ cannot be a bad edge (by part (b) of Lemma~\ref{M4bad}). So, $e$ has to be added to $M_4$ in the augment phase for it to be a bad leaf edge in future.
 \begin{itemize}
  \item If $e$ was added to $M_4$ alone when revealed, then it must have neighbors $e_1$ and $e_2$ such that both of them do not belong to $M_4$. Then, they must have had neighboring edges $e_1'$ and $e_2'$ respectively which belonged to $M_4$ (at some stage). Suppose $e_1''$ (and/or $e_2''$) switches $e_1'$ (and/or $e_2'$ respectively) out of $M_4$, then $e_1''$ (and/or $e_2''$ respectively) cannot be a bad edge (by part (b) of Lemma~\ref{M4bad}). Otherwise, the Lemma holds due to part (a) of Lemma~\ref{M4bad}.
  \item If $e$ was added to two matchings ($M_4$ being one of them) when it was revealed, and finally has only one internal neighboring edge $e_1$, then $e_1$ will have a neighboring edge $e_2$ on its other end point. Either $e_2$ belongs to $M_4$ or its neighboring edge $e_2'$ on other end point belongs to $M_4$. The lemma holds if finally $e_2'$ belongs to $M_4$ (by part (a) of Lemma~\ref{M4bad}) or if finally the neighboring edge $e_2''$ of $e_2'$ belongs to $M_4$ (by part (b) of Lemma~\ref{M4bad}). (The proof for this case will also work for the case when $e$ was revealed first as a single disconnected edge, and then $e_1$ was revealed on one of its end points.)
  \item If $e$ was added to two or three matchings ($M_4$ being one of them) when it was revealed, and finally has two internal neighboring edges $e_1$ and $e_2$, then $e_1$ and $e_2$ must have neighboring edges $e_1'$ and $e_2'$ respectively which belong to $M_4$ (at some stage). Suppose $e_1''$ (and/or $e_2''$) switches $e_1'$ (and/or $e_2'$ respectively) out of $M_4$, then $e_1''$ (and/or $e_2''$ respectively) cannot be a bad edge (by part (b) of Lemma~\ref{M4bad}). Otherwise, the Lemma holds due to part (a) of Lemma~\ref{M4bad}.
 \end{itemize}
\item Let $e_1$ and $e_2$ be two bad internal edges which do not belong to $M_4$. Then, they must have had neighboring edges $e_1'$ and $e_2'$ respectively which belonged to $M_4$ (at some stage). Suppose $e_1''$ (and/or $e_2''$) switches $e_1'$ (and/or $e_2'$ respectively) out of $M_4$, then $e_1''$ (and/or $e_2''$ respectively) cannot be a bad edge (by part (b) of Lemma~\ref{M4bad}). Otherwise, the Lemma holds due to part (a) of Lemma~\ref{M4bad}.
\end{enumerate}

\end{proof}

\end{appendix}

\end{document}